\documentclass[twoside,a4paper,11pt]{amsart}

\usepackage[english]{babel}
\usepackage{amsmath}
\usepackage{amssymb}
\usepackage{url}
\usepackage{amsthm}
\usepackage{amsxtra}
\usepackage{mathrsfs}
\usepackage{fancyhdr}
\usepackage{hyperref}
\usepackage{enumerate}

\newtheorem{theorem}{Theorem}
\newtheorem{lemma}{Lemma}[section]
\newtheorem{proposition}[lemma]{Proposition}

\theoremstyle{definition}
\newtheorem{remark}[lemma]{Remark}

\numberwithin{equation}{section}

\newcommand{\pd}[2]{\frac{\partial {#1}}{\partial {#2}}}
\newcommand{\beq}{\begin{equation}}
\newcommand{\eeq}{\end{equation}}
\newcommand{\be}{\begin{equation*}}
\newcommand{\ee}{\end{equation*}}

\newcommand{\RE}{\mathbb R}
\newcommand{\CO}{\mathbb C}

\newcommand{\FF}{\mathcal F}

\newcommand{\de}{\delta}

\newcommand{\la}{\lambda}
\newcommand{\ome}{\omega}

\renewcommand{\Im}{\operatorname{Im}\,}

\newcommand{\PV}{\operatorname{PV}}

\title[]{Time dependent delta-prime interactions in dimension one}

\author[]{Claudio Cacciapuoti}
\address[]{DiSAT, Sezione di Matematica, Universit\`a dell'Insubria, via Valleggio 11, 22100
Como, Italy}
\email{claudio.cacciapuoti@uninsubria.it}

\author[]{Andrea Mantile}
\address[]{Laboratoire de Math\'{e}matiques, Universit\'{e} de Reims -
FR3399 CNRS, Moulin de la Housse BP 1039, 51687 Reims, France}
\email{andrea.mantile@univ-reims.fr}

\author[]{Andrea Posilicano}
\address[]{DiSAT, Sezione di Matematica, Universit\`a dell'Insubria, via Valleggio 11, 22100
Como, Italy}
\email{andrea.posilicano@uninsubria.it}

\thanks{
The authors acknowledge the support of the FIR 2013 project ``Condensed Matter in Mathematical Physics'', Ministry of University and
Research of Italian Republic  (code RBFR13WAET)}

\begin{document}


\begin{abstract} We solve the Cauchy problem for the Schr\"odinger equation corresponding to the family of  Hamiltonians $H_{\gamma(t)}$ in $L^{2}(\mathbb{R})$ 
which describes a $\delta'$-interaction with time-dependent strength $1/\gamma(t)$. We prove that the strong solution of such a Cauchy problem exits whenever the map $t\mapsto\gamma(t)$ belongs to the fractional Sobolev space $H^{3/4}(\mathbb{R})$, thus weakening the hypotheses which would be required by the known general abstract results. The solution is expressed in terms of the free evolution and the solution of a Volterra integral equation.   
\end{abstract}

\maketitle

\begin{footnotesize}

\end{footnotesize}

\section{Introduction}

In this paper we address the generation problem for the family of time-dependent  
Hamiltonians $H_{\gamma(t)}$, where $H_{\gamma(t)}$, for any fixed real $t$, denotes the self-adjoint operator in $L^{2}(\mathbb{R})$ describing a $\delta'$-interaction of strength $1/\gamma(t)$ (see \cite{seba:1986,gesztesy:1987ab}, \cite[Chapter I.4]{Albeverio:2005vf} and references therein). 

Most of the literature on time dependent point interactions focuses on perturbations of the free dynamics of  the form of a Dirac's delta time dependent potential. In  dimension three, time dependent $\delta$-interactions were studied  in \cite{Sayapova:1983wu,Yafaev:1984ux} and in  \cite{Correggi:2005ev} in  relation with the problem of ionization under periodic perturbations, see also \cite{Correggi:2005gu}. In dimension two, very recently, the problem of the well-posedness was studied in \cite{carlone:2016}. In dimension one, this kind of non-autonomous Hamiltonians were analyzed in \cite{Hmidi:2009kg}, see also \cite{neidhardt2009linear}. 

It is well known that in dimension one the family of point perturbations of the Laplacian is richer than in dimension two and three, and includes $\delta$ and $\delta'$ perturbations, as well as their combinations. In this paper we focus   attention  on the topical case of a time dependent  $\delta'$-interaction. 

We remark that time-dependent $\delta$-interactions have a non-linear counterpart, see, e.g., \cite{Adami:2003hm,Adami:2004bp,Adami:2013bt} in dimension three, and \cite{Adami:1999dm,Adami:2001bt}. More recently a systematic study of the blow-up in the one dimensional case was started in \cite{holmer:2015}.   In dimension one, in particular, such models find applications to the propagation of optical waves in Kerr media, or one-dimensional many body systems, see, e.g., \cite{Dror:2011kr,Malomed:1993bg,hennig:1994,molina:2002} and references therein. The problem of the derivation of non-linear $\delta$-interactions from scaled regular dynamics was recently studied in \cite{Cacciapuoti:2014gt} in dimension one, and  \cite{cacciapuoti:2015ab} in dimension three, see also \cite{cacciapuoti:2015}. 

Several results discussed in the present paper set the ground for  the definition of  non-linear point interactions of $\delta'$-type and for the study of the problem of  their derivation from   scaled regular dynamics. 

We recall that  the definition of $H_\gamma$ is given by the theory of self-adjoint extensions of the symmetric operator $$H^{\circ}=-\Delta\equiv -\frac{d^{2}%
}{dx^{2}}\,,\qquad D(H^{\circ})=C_{0}^{\infty}(\mathbb{R}\backslash\{  0\})\,,$$
and, for any real $\gamma$, reads as follows:
\begin{equation}\label{defH}
H_{\gamma}\psi(x)=-\frac{d^{2}\psi}{dx^{2}}(x)\,,\qquad x\not=0\,,
\end{equation}
\begin{equation}\label{defD}
D(H_{\gamma})=\{  \psi\in L^{2}(\mathbb{R}):\psi=\phi+q\eta,\ \phi\in X^2,\ q\in\mathbb{C},\ \phi'(0)=\gamma q  \}\,,
\end{equation}
where $\eta(x):=\frac{1}2\,\text{\rm sgn}(x)$ and for any $\nu\geq0$ we defined $X^\nu$ as the space of tempered distributions with Fourier transform in $L^2(\RE, |k|^{2\nu}dk)$. 

We remark that  if $f \in X^\nu$, then its Fourier transform might be a distribution as well. Moreover, for $\nu = m +\sigma$, with $m$ integer and $1/2<\sigma\leq1$, if $f\in X^\nu $ then $f \in C^{m}(\RE)$, see Prop. \ref{p:Xnu} below. Hence, $\phi$ in $D(H_{\gamma})$ is  a  $C^1(\RE)$  function and $\phi'(0)$ in  the boundary condition     is well defined. 

The action of the operator $H_\gamma$ can be understood also by exploiting the decomposition $\psi = \phi + q\eta$: this leads to 
\begin{equation}\label{action}
H_\gamma \psi(x) = -\phi''(x), \qquad x\in \RE. 
\end{equation}

When $\gamma(t)$ is assigned as a real valued function of time, the domain
$D(H_{\gamma(t)})$ changes in time with the boundary condition $\phi'(0)=\gamma(t) q$. In contrast, the quadratic form corresponding to $H_{\gamma}$ is given by 
\[
Q_{\gamma}(\psi)=\|\phi'\|^{2}+\gamma |q|^2\,,
\]
\[
D(Q_{\gamma})=\{  \psi\in L^{2}(\mathbb{R}):\psi=\phi+q\eta,\  \phi\in X^1,\ q\in\mathbb{C}\,\},
\]
and so $Q_{\gamma(t)}$ has a  time-independent domain. Thus,  by the abstract results in \cite{K} and \cite{neidhardt2009linear}, assuming that the map $t\mapsto\gamma(t)$ is differentiable, there exists an unitary propagator  $U_{t,s}$ in $L^{2}(\mathbb{R})$, continuously mapping $D(H_{\gamma(s)})$ onto $D(H_{\gamma(t)})$, such that $\psi(t):=U_{t,0}\psi_{0}$ is the (strong) solution of the Cauchy problem 
\begin{equation}\label{cauchy}
\begin{cases}
i\frac{d}{dt}\psi(t)=H_{\gamma(t)}\psi(t)\\
\psi(0)=\psi_{0}\in D(H_{\gamma(0)})\,.%
\end{cases}
\end{equation}
However, as the case of time-dependent self-adjoint extensions $H_{\alpha(t)}$ (corresponding to $\delta$-in\-ter\-ac\-tions) studied in \cite{Hmidi:2009kg} suggests, 
the quite explicit knowledge of the action and operator domain of $H_{\gamma}$ should allow to solve  the Cauchy problem \eqref{cauchy} under weaker regularity conditions on $\gamma(t)$.  Indeed, as we show in this paper, this is the case and problem $\eqref{cauchy}$ has an unique strong solution whenever the map $t\mapsto\gamma(t)$ is in the fractional Sobolev space $H^{3/4}(\mathbb{R})$, a condition weaker than the differentiability hypotheses required in \cite{K} and \cite{neidhardt2009linear}.  Such a $H^{3/4}$ hypothesis is the same required in the paper   \cite{Hmidi:2009kg} in order to guarantee that  the Cauchy problem for the family $H_{\alpha(t)}$ has a strong solution. However, in contrast to \cite{Hmidi:2009kg}, here we make use neither of sophisticated analytic tools (paraproducts) nor of abstract generation theorems (as the ones provided in \cite{K} and \cite{kato:1970ab}); instead, following the same strategy as in the paper \cite{DellAntonio:1996tj}, we apply a more direct approach which exploits the definitions \eqref{defH} and \eqref{defD}, providing a relatively explicit expression for the solution of \eqref{cauchy} with initial datum $\psi_{0}=\phi_{0}+q_{0}\eta$ in $D(H_{\gamma(0)})$:
\begin{equation}\label{psi}
\psi(t) = \phi(t) +  q(t) \eta
\end{equation} 
with 
\begin{equation}\label{phi}
\phi(t)=e^{it\Delta}\phi_{0}-\int_{0}^{t}ds\,\dot q(s)e^{i(t-s)\Delta}\eta\,,
\end{equation}
where $t\mapsto q(t)$ solves the Volterra-type integral equation 
\begin{equation}\label{q}
q(t)=f_{0}(t)-\sqrt{\frac{4i}{\pi}}\int_{0}^{t}ds\,\frac{\gamma(s)q(s)}{\sqrt{t-s}}
\end{equation}
and the source term $f_{0}$ is defined as  
\begin{equation} \label{f0}
 f_0(t):=q_{0}+\sqrt{\frac{4i}{\pi}}\int_{0}^{t}ds\,\frac{(e^{is\Delta}\phi_{0})^{'}(0)}{\sqrt{t-s}}\,.
\end{equation}


We shall prove the following well-posedness result:
\begin{theorem}\label{t:mainth}
Let $T>0$,  $\gamma \in H^{3/4}(0,T)$, and set $\gamma_0=\gamma(0)$.  Let $\psi_0 = \phi_0 + q_0 \eta \in D(H_{\gamma_0})$. Then for any $t\in[0,T]$ there exists a unique strong solution of the Cauchy problem \eqref{cauchy}  given by $\psi(t) = \phi(t) + q(t) \eta$ as in Eqs. \eqref{phi} - \eqref{f0}. 
Moreover the map  $t\mapsto H_{\gamma(t)} \psi(t)$ 
belongs to $C([0,T],L^2(\RE))$.
\end{theorem}

We briefly discuss the heuristic derivation of the solution. The solution of the Schr\"odinger equation with $H_\gamma$ as  Hamiltonian satisfies the distributional equation 
\begin{equation}
\label{schreq}
i\pd{}{t} \psi(t) = -\psi''(t) + q(t) \delta'_0 ,
\end{equation}
where $\delta_0' $ is the first derivative of the Dirac delta-distribution.  Let us assume, in the first part of this discussion, that the source term  $q(t)$ were an assigned function. Since $\eta'' = \delta_0'$,  it is natural to seek for solutions  of the form \eqref{psi}. Setting  $\psi(t) = \phi(t) +  q(t) \eta$  in Eq. \eqref{schreq} gives the equation for $\phi(t)$ 
\[
i\pd{}{t} \phi(t) = -\phi''(t)  - i\dot q(t) \eta.
\]
Eq. \eqref{phi} follows directly from the  Duhamel's formula. Indeed by integration by parts, see Section \ref{ss:proof} (in particular Eqs. \eqref{house} and \eqref{nota}),  one obtains the following equation for $\psi (t)$
\begin{equation}\label{duapsi}
\psi(t)=e^{it\Delta}\psi_{0}- i\int_{0}^{t}ds\, q(s)e^{i(t-s)\Delta} \delta_0'\,.
\end{equation}
This can be  understood as the Duhamel's formula applied to Eq. \eqref{schreq}. 

The equation for $q(t)$ is obtained by imposing the boundary condition $\phi'(0)=\gamma q $, using Eq. \eqref{phi} to compute the l.h.s. in the boundary condition. We postpone the details of the calculation  to Section \ref{ss:proof}. Here we note that the boundary condition turns the flow associated to Eq. \eqref{schreq} into a unitary flow. In fact  one can show that 
\[\frac{d}{dt} \|\psi(t)\|^2 = 2 \Im \bar q(t) \phi'(0,t).\]
Hence, if the boundary condition is satisfied, one has $\frac{d}{dt} \|\psi(t)\|=0$.  \\ 

We remark that a function $\psi\in D(H_\gamma)$ can be written as the sum of a regular and singular part with both functions in $L^2$ by introducing a regularization parameter $\lambda$. More precisely, define 
\[G_\lambda(x) : = - \frac{e^{-\sqrt\lambda|x|}}{2\sqrt\lambda} \qquad \lambda >0.\]
The function $G_\lambda$ is the solution of the distributional equation $G_\lambda'' = \delta_0 +\lambda G_\lambda$. The  domain $D(H_\gamma)$ can be rewritten as 
\[
D(H_{\gamma})=\Big\{  \psi\in L^{2}(\mathbb{R}):\psi=\phi_\lambda+q G_\lambda',\ \phi_\lambda\in H^2(\RE),\ q\in\mathbb{C},\ \phi_\lambda'(0)=\Big(\gamma +\frac{\sqrt\la}{2}\Big)q \Big\},
\]
and the action of $H_\gamma$ can be understood by the identity 
\[(H_\gamma +\lambda )\psi(x) = -\phi_\lambda''(x) +\lambda\phi_\lambda(x), \qquad x\in\RE,\]
see, e.g., \cite{Albeverio:2005vf}. Eq. \eqref{action} is obtained by taking $\lambda \to 0$. 

We note that the charge equation \eqref{q} does not depend on $\lambda$, it is easy to see that 
\[
f_0(t)=\sqrt{\frac{4i}{\pi}}\int_{0}^{t}ds\,\frac{(e^{is\Delta}\psi_{0})^{'}(0)}{\sqrt{t-s}},
\]
see Eqs. \eqref{Uetaprimo} and \eqref{i} below. The equation for the regular part $\phi_\lambda$, instead,  does  involve the regularization parameter, precisely 
\[\phi_\lambda(t)=e^{it\Delta}\phi_{\lambda,0}-\int_{0}^{t}ds\,(\dot q(s)+\lambda q(s))e^{i(t-s)\Delta}G'_\lambda. 
\]
We note that, even if the regularization   would avoid few issues with  convolutions and Fourier transforms, which must otherwise interpreted in distributional sense,  it makes formulae more involved and introduces an unnecessary parameter. For this reasons we decided to avoid it. \\

The paper consists of one additional  section in which we prove Theorem \ref{t:mainth}.  

\section{Proof of  Theorem \ref{t:mainth}}
\subsection{Notation and preliminaries} In what follows $C$ denotes a generic  positive  constant whose value may change from line to line. 
  
We denote by $\hat \psi$ the spatial Fourier transform of $\psi$ 
\[
\hat \psi(k) =  \int_{\RE} \,dx \, e^{-ikx} \psi(x)\ .
\]

The time-Fourier transform  of $f$ is denoted by $\FF f$ and defined as
\[
\FF f (\ome) = 
\int_{\RE} \,dt\, e^{-i\ome t} f(t)\ .
\]
With these definitions the  Fourier transform of the convolution is \[(\widehat{\psi*\phi})(k) =  \hat \psi(k) \hat \phi(k),\]and similarly for the time-Fourier transform.  

In the following, we denote by $U(t)$ the  free unitary group $e^{i\Delta t}$, we recall that its explicit expression is   given by  
\[
U(t)\psi(x)  = \int_\RE dy\,  \frac{e^{\frac{i(x-y)^2}{4t}}}{\sqrt{4\pi i t}} \psi(y) ,
\]
which in Fourier transform reads
\[\widehat{U(t)\psi}(k) = e^{-i k^2 t } \hat \psi(k).  \] 

\begin{proposition}\label{p:Xnu}
For $\nu = m +\sigma$, with $m$ integer and $1/2<\sigma\leq 1$, it results  $X^\nu \subset C^{m}(\RE)$.
\end{proposition}
\begin{proof}
In Fourier transform 
\[f^{(m)}(x) - f^{(m)}(y) = \frac1{2\pi} \int_\RE dk \, (ik)^m(e^{ik x} - e^{ik y}) \hat f(k).\]
We note that 
\begin{equation}\label{y1}\begin{aligned}
\left| \int_{|k|<1} dk \, (ik)^m(e^{ik x} - e^{ik y}) \hat f(k) \right| \leq & C |x-y|^\sigma  \int_{|k|<1} dk  \, |k|^{m+\sigma} |\hat f(k)| \\ 
 \leq & C |x-y|^\sigma  \|\hat f\|_{L^2(\RE,|k|^{2\nu}dk)}  .
\end{aligned}
\end{equation}
Moreover
\begin{equation}\label{y2}\begin{aligned}
\left| \int_{|k|>1} dk \, (ik)^m(e^{ik x} - e^{ik y}) \hat f(k) \right| \leq & C \int_{|k|>1}dk \, |k|^{m} |\hat f(k)| \\ 
  \leq  &  C \left(\int_{|k|>1} \frac{dk}{|k|^{2\sigma}}\right)^{\frac12} \|\hat f\|_{L^2(\RE,|k|^{2\nu}dk)}.
\end{aligned}\end{equation}
Then the continuity of $f^{(m)}$  follows from the bounds \eqref{y1} and \eqref{y2},  and  the dominated convergence theorem.
\end{proof}

We will make use of fractional Sobolev spaces, for this reason we recall few definitions.  For any $-\infty\leq a <b\leq +\infty$ and $\nu\in(0,1)$,  we set   
\[
[f]_{H^\nu(a,b)} :=  \left(\int_{[a,b]^2 } d s d s' \frac{ | f(s) -f(s') |^2  }{ |s-s'|^{1+2\nu} }\right)^{1/2},
\]
which is sometimes  referred to as Gagliardo (semi)norm of $f$. The space $H^{\nu}(a,b)$, for $-\infty\leq a <b\leq +\infty$ and $\nu\in(0,1)$,  is the space of functions for which the norm \[\|f\|_{H^\nu(a,b)} = \|f\|_{L^2(a,b)} + [f]_{H^\nu(a,b)}\] is finite.  To define the space $H^\nu(a,b)$  for   $\nu>1$ not integer, one sets $\nu = m + \sigma$, where $m$ is an integer and $\sigma\in(0,1)$.  Then $H^\nu(a,b)$ is the space of functions such that $f\in H^m(a,b)$ and  $f^{(m)}\in H^\sigma(a,b)$.

\begin{remark}\label{r:XH}Note that,    for $\nu\in(0,1)$ there exists a constant $C_\nu$ such that 
\[[f]_{H^\nu(\RE)} = C_\nu \|\FF f\|_{L^2(\RE,|\omega|^{2\nu}d \omega)},\]
for any $f\in X^\nu$,  this is a direct consequence of Plancherel's theorem  (see \cite{bahouri}, Proposition 1.37). This identity, together with Prop. \ref{p:Xnu} implies that, for all $\nu>1/2$, and $a$ and $b$ finite, if  $f\in X^\nu $ then  $f\in H^\nu(a,b)$, and, consequently, it belongs to  $H^\mu(a,b)$ for all $0\leq \mu \leq \nu$. 
 Also, if $f\in L^2(a,b)$ and $f\in X^\nu$, then  $f\in H^\nu(a,b)$, and, consequently, in $H^\mu(a,b)$ for all $0\leq \mu \leq \nu$. 
\end{remark}

We recall that, for $-\infty\leq a <b\leq +\infty$,  the space  $L^2(a,b)$ can be identified with $H^0(a,b)$, and $L^2(\RE)$ can be identified with $X^0$. 

For the norms, we shall  use the notation $\|\cdot \| = \|\cdot\|_{L^2(\RE)}$. 

We denote by $I$ the operator 
\begin{equation}\label{I}
I f(t) = \frac{1}{\sqrt\pi} \int_0^t\, ds\,  \frac{f(s)}{\sqrt{t-s}}.  
\end{equation}

We shall use the following results which establishes the regularization properties of the operator $I$. 
\begin{lemma}\label{l:reg} Let  $\nu \geq0$ and $T>0$. Assume that   $f \in X^\nu$   and has  support in $[0,T]$, then $If\in X^{\nu+1/2}$.
\end{lemma}
\begin{proof}
%
The integral kernel 
\[A(t) = \frac1{\sqrt\pi} \frac{ \Theta(t)}{\sqrt t}, \]where $\Theta$ is the Heaviside function, is a tempered distribution and 
\[\FF A(\omega) = \frac{1}{\sqrt{|\omega|}} \left(\frac{\sqrt i}{2} \Theta(\omega) + \frac{1}{\sqrt 2}\left(\Theta (-\omega)+i\Theta (\omega)\right)\right).\]
Let  $f\in X^\nu$. The convolution of $A$ and $f$, $If = A* f$,  is a tempered distributions and $\FF If= \FF A\FF f$, see, e.g., \cite[Th. 14.25]{duistermaat}. Then 
\[\||\cdot|^{\nu+1/2}\FF If\| \leq C \||\cdot|^{\nu}\FF f\|.\]
\end{proof}
We recall the following   technical lemma. 
 \begin{lemma}\label{l:prolong}
 Let $-\infty<a<b<\infty$ and let $f\in H^\nu(a,b)$ with $\nu\geq0$. Define 
 \[\tilde f (s) = \left\{ \begin{aligned} &f(s) \quad && \text{if} \quad s \in [a, b] \\ 
 & 0 && \text{otherwise}  \end{aligned}\right. \]
 \begin{enumerate}[i)]
\item If  $0\leq\nu <1/2$, then $\tilde f\in H^\nu(\RE)$. 
\item  If  $1/2 <\nu <3/2$ and $f(a)=f(b)=0$, then $\tilde f\in H^\nu(\RE)$ .
 \end{enumerate}
\end{lemma}
For a proof  see for example \cite[Th. 11.4]{LM12}, see also \cite[Th. III.3.2]{strichartz:1967}.  

We shall also  use the following:
\begin{proposition}\label{p:sobprod}
Let  $\mu > 1/2$ and  $0\leq \nu \leq \mu$. If $g\in H^\mu(a,b)$ and $f\in H^{\nu}(a,b)$ then $fg\in H^{\nu}(a,b)$. 
\end{proposition}
For the proof we refer to \cite{strichartz:1967}.

\subsection{Well-posedness  of the charge equation}
In this section we study the well-posedness  of the charge equation  \eqref{q}. 

We start with  the following lemma which gives the regularity properties of the inhomogeneous term in Eq. \eqref{q}.  
\begin{lemma}\label{l:phi0}
Let $\phi_0 \in  X^2$, then $(U(\cdot)\phi_0)'(0) \in X^{3/4}$. 
\end{lemma}
\begin{proof}
Since  $\phi_0' \in L^2(\RE)$, one has that the distributional identity  
\[
(U(t)\phi_0)'(x) = \int_\RE dy\,  \frac{e^{\frac{i(x-y)^2}{4t}}}{\sqrt{4\pi i t}} \phi_0'(y)
\]
shows that $(U(t)\phi_0)' \in L^2(\RE)$.  By using the Fourier transform one has that 
\[
(U(t)\phi_0)'(0) = \frac{1}{2\pi} \int_\RE dk\, e^{-ik^2t} \widehat{\phi'_0}(k). 
\]
By splitting  the integral in $dk$ for  $k>0$ and $k<0$, and by using the change of variables $k=\sqrt\ome$ for $k>0$ and $k = -\sqrt \ome$ for $k<0$, it follows that  
\[
(U(t)\phi_0)'(0) = \frac{i}{4\pi} \int_0^\infty \frac{d\ome}{\sqrt \omega}\,  e^{-i\ome t} (\widehat{\phi_0'}(\sqrt\ome) + \widehat{\phi_0'}(-\sqrt\ome)) .
\]
Hence 
\[
\FF\left((U(\cdot)\phi_0)'(0) \right)(\ome)=  \frac{i}{2\sqrt \omega} \Theta(-\ome) (\widehat{\phi_0'}(\sqrt{-\ome}) + \widehat{\phi_0'}(-\sqrt{-\ome})) ,
\]
where $\Theta$ denotes the Heaviside function.  
To prove that $\FF\left((U(\cdot)\phi_0)'(0) \right) \in L^2(\RE,|\ome|^{\frac32}d\ome)$ it is enough to note that 
\[\||\cdot|^{\frac34}\FF\left((U(\cdot)\phi_0)'(0) \right)\| \leq  C \||\cdot|\widehat{\phi'_0}\| =  C \||\cdot|^2\hat{\phi_0}\|,
\]
where we used the change of variables $k^2 = \omega$.  
\end{proof}

We are now ready to prove the main result of this section. 
\begin{lemma}
Let $T> 0$,  $\gamma \in H^{3/4}(0,T)$, and set $\gamma_0=\gamma(0)$.  Let $\psi_0 = \phi_0 + q_0 \eta \in D(H_{\gamma_0})$.  Then Eq.  \eqref{q} admits a unique solution $q\in  H^{5/4}(0,T)$. 
\end{lemma}
\begin{proof}
We split the proof in two steps: first we prove  that  there exists a unique solution $q\in L^2(0,T)$, then, by  a  bootstrap argument, we  show that such solution belongs to $ H^{5/4}(0,T)$. 

We start by step 1. We use several results from the monograph \cite{gripenberg:1990}. We set 
\[
k(t,s) =
 \sqrt{\frac{4 i }{\pi}} \frac{\gamma(s)}{\sqrt{t-s}} 
\]
and rewrite the equation as 
\begin{equation}\label{q2}
q(t) = f_0(t) - \int_0^t \,ds\, k(t,s) q(s).  
\end{equation}
This is a linear nonconvolution Volterra equation to which we can apply the results in \cite[Ch. 9]{gripenberg:1990}.  We start by noticing that for any finite interval $J\subset\RE^+$,   $k(t,s)$ is a Volterra kernel of type $L^2$, more precisely 
\[
|||k|||_{L^2(J)} := \sup_{\begin{subarray}{l}
        \|h\|_{L^2(J)}\leq1 \\  \|g\|_{L^2(J)}\leq1      \end{subarray}
     } \int_J\int_J \,ds\,dt\, |h(t) k(t,s) g(s)|  \leq C |J|^{1/2}   \|\gamma\|_{L^\infty(J)}.
\]
Hence the interval $[0,T]$ can be divided into finitely many subintervals  $J_i$ such that  $|||k|||_{L^2(J_i)} <1$ on each $J_i$, and, as a consequence of  Cor. 9.3.14 in \cite{gripenberg:1990}, one has that $k$ has a resolvent of type $L^2$ on $[0,T]$. By applying Th. 9.3.6 of \cite{gripenberg:1990}, we conclude that Eq. \eqref{q2} has a unique solution in $L^2(0,T)$. 

We can now proceed to the second step of the proof, which consists in showing that such a solution belongs to $H^{5/4}(0,T)$. By Lemma \ref{l:phi0} and Rem. \ref{r:XH}, one has $(U(\cdot)\phi_0)'(0)\in H^\nu(0,T)$ for all $0\leq \nu \leq 3/4$. We set 
\[
Q(t) = q(t) - q_0 \qquad \text{and} \qquad  F(t) =  \sqrt{4 i } ((U(t)\phi_0)'(0) - \gamma(t) q(t)) \qquad t\in[0,T].
\]
We denote by $\tilde Q$ the function obtained by prolonging $Q$ to zero outside $[0,T]$ and   remark that the claim $\tilde Q\in X^\nu$ implies $Q\in H^\mu(0,T)$ for all $0\leq \mu \leq \nu$, see Rem. \ref{r:XH}, therefore  $q\in H^\mu(0,T)$. 

By the charge  equation \eqref{q},  the identity $Q = I F$ holds true for a.a. $t \in [0,T]$, here $I$ is the operator defined in \eqref{I}. Since, by Prop. \ref{p:sobprod}, $F\in L^2(0,T)$ we  can define $\tilde F \in L^2(\RE)$ by extending it to zero. Then, by Lemma \ref{l:reg}, $\tilde Q = I \tilde F \in X^{1/2}$, hence, $Q\in H^{1/4}(0,T)$ and $q\in H^{1/4}(0,T)$.

We can repeat the argument. We start  with the observation that now we know that $F\in H^{1/4}(0,T)$ and conclude that $q\in H^{3/4}(0,T)$. Here we use Lemma \ref{l:prolong}-$i)$ to claim that $\tilde F\in H^{1/4}(\RE)$ which in turn implies $\tilde F\in X^{1/4}$.

To conclude the proof we must slightly adjust the argument above. So far we have proved that $F \in H^{3/4}(0,T)$, moreover we know that $F(0) = 0$, because the boundary condition $\phi_0'(0) = \gamma_0q_0$ holds true by assumption. Define  $F^s : [0,2T]\to \CO$ by reflection of $F$ about  $t=T$. We have that $F^s(0) = F^s(2T)=0$. We define $\tilde F^s :\RE\to\CO$ by extending $F^s$ to zero and use   Lemma \ref{l:prolong}-$ii)$ to claim that $\tilde F^s\in  H^{3/4}(\RE)$, and, consequently, $\tilde F\in X^{3/4}$.  Applying again Lemma \ref{l:reg}  we conclude that $q\in  H^{5/4}(0,T)$. 
\end{proof}

\subsection{Proof of Theorem \ref{t:mainth}} \label{ss:proof}
The function $\phi(t)$ defined by Eq. \eqref{phi}  exists and is unique for all $t\in [0,T]$. Next we prove that $\phi(t) \in X^2$.  Let us rewrite Eq. \eqref{phi} as 
\[
\phi(t) = U(t)\phi_0 + \tilde \phi(t) ,
\]
where we set 
\begin{equation}\label{tildephi1}
\tilde \phi(t) = - \int_0^t ds \, \dot q(s) U(t-s)\eta .
\end{equation}
One has that $U(t)\phi_0 \in X^2$, because  $\|\widehat{U(t)\phi_0}\|_{L^2(\RE,|k|^4dk)} =\|\hat \phi_0\|_{L^2(\RE,|k|^4dk)} $.  

We are left to prove that $\tilde \phi \in X^2$. We recall that the Fourier transform of $\eta$ is the distribution  $-i \PV \frac1{k}$ (where $\PV$ stands for principal value).   We have that 
\begin{equation}\begin{aligned}\label{tildephi2}
\|\hat{\tilde \phi}(t)\|^2_{L^2(\RE, |k^4|dk)} = &  \frac{1}{2\pi}\int_\RE dk \, k^2 \left| \int_0^t ds\,  e^{-ik^2 (t-s)}  \dot q(s)\right|^2  \\ 
= &  \frac{1}{2\pi}\int_0^\infty d\omega \, \ome^{\frac12} \left| \int_0^t  ds\,  e^{i\ome s}  \dot q(s)\right|^2 \leq C\|\dot q\|_{H^{1/4}(0,T)}.
\end{aligned}\end{equation}
Here, the inequality follows from the same argument used in the proof of Prop. 3.3 in \cite{cacciapuoti:2015ab}. 

Next we prove that $\psi(t) = \phi(t) + q(t)\eta \in L^2(\RE)$. Since $\phi(t)\in C^1(\RE)$, see Prop. \ref{p:Xnu},  and $\eta$ is  bounded, $\psi(t) \in L^2_{loc}(\RE)$. Hence, it is enough to prove that $(1-\chi) \psi(t) \in L^2(\RE)$, where $\chi$ is the characteristic function of the interval $[-1,1]$.   In the definition of $\phi(t)$, see Eq. \eqref{phi}, we use the identity 
\[
\int_0^t ds \, \dot q(s) U(t-s)\eta  = q(t) \eta - q_0 U(t)\eta - \int_0^t ds \,  q(s) \pd{}{s}  U(t-s)\eta, 
\]
which gives 
\begin{equation}\label{house}\psi(t) = U(t) \psi_0  + \int_0^t ds \,  q(s) \pd{}{s}  U(t-s)\eta.
 \end{equation}
Since $U(t)\psi_0\in L^2(\RE)$ we are left to prove that the second term at the r.h.s., times the function $(1-\chi)$,  is in $L^2(\RE)$ as well. 
 We note that 
\begin{equation}\label{Ueta}\begin{aligned}
(U(t)\eta )(x) = & \int_\RE dy \, \frac{e^{i\frac{(x-y)^2}{4t}}}{\sqrt{4\pi i t} } \eta(y) \\ 
  = &\frac12  \frac{1}{\sqrt{4\pi i t} } \left( \int_{-\infty}^x dy \, e^{i\frac{y^2}{4t}} -   \int_x^{\infty} dy \, e^{i\frac{y^2}{4t}} \right).
\end{aligned}
\end{equation}
From which we get 
\[
\pd{}{t}  (U(t)\eta)(x) =  -\frac12 \frac{1}{\sqrt{4\pi i}}  \frac{x}{t^{3/2}} e^{i\frac{x^2}{4t}} = - \sqrt{\frac{i}{\pi}} \frac{\sqrt t}{x} \frac{d}{dt} e^{i\frac{x^2}{4t}} . 
\]
We remark that the first equality can be understood in distributional sense as
\begin{equation}\label{nota}
\pd{}{t}  (U(t)\eta)=  i (U(t)\eta)''=iU(t)\eta'' = i U(t)\delta_0',
\end{equation}
from which one deduces that Eq. \eqref{house} is equivalent to Eq. \eqref{duapsi}. 

Which gives 
\[\begin{aligned} & \int_0^t ds \,  q(s) \pd{}{t}  (U(t-s)\eta)(x) \\ 
 = &  \sqrt{\frac{i}{\pi}} \frac1x  \int_0^t ds \,  q(s)  \sqrt{t-s} \frac{d}{ds} e^{i\frac{x^2}{4(t-s)}}  \\ 
 = & 
 \sqrt{\frac{i}{\pi}} \frac1x  \left( -   q_0  \sqrt{t} e^{i\frac{x^2}{4t}}  -  \int_0^t ds \,  \dot q(s) \,  \sqrt{t-s} \, e^{i\frac{x^2}{4(t-s)}} +  \frac12 \int_0^t ds \,  \frac{q(s)}{\sqrt{t-s}}e^{i\frac{x^2}{4(t-s)}} \right)  . \end{aligned} \]
We gained a factor $1/x$ which gives the bound 
\[  
\left\| (1-\chi) \int_0^t ds \,  q(s) \pd{}{t}  U(t-s)\eta  \right\| \leq C (\|q\|_{L^\infty(0,T)} + \|\dot q\|_{L^1(0,T)}) \leq C \qquad t\in[0,T]. 
\]
Next we prove that the boundary condition $ \phi'(0) = \gamma(t) q$ holds true for all $t\in[0,T]$. From Eq. \eqref{Ueta} we obtain 
\begin{equation} \label{Uetaprimo}
(U(t)\eta )'(0) = \frac{1}{\sqrt{4\pi i t} } . \end{equation}
Hence
\[
\phi'(0,t) = (U(t)\phi_0)'(0)  - \int_0^t ds \, \frac{1}{\sqrt{4\pi i (t-s)} }\, \dot q(s).
\]
We apply the operator $I$, defined in \eqref{I}, and use the charge equation \eqref{q} to  obtain 
\[
(I \phi'(0,\cdot))(t) = (I (U(\cdot)\phi_0)'(0))(t)  -   \frac{1}{\sqrt{4 i } }( q(t) - q_0) = (I\gamma q)(t),
\]
which imply  the boundary condition. Here we used the identities 
\begin{equation}\label{i}
I (\pi (\cdot))^{-1/2}(t) = \frac1{\sqrt{\pi}}\int_0^t ds\, \frac{1}{\sqrt{t-s}}\frac{1}{\sqrt{\pi s}} = 1 \qquad \text{and} \qquad I^2 f(t) = \int_0^t ds\, f(s).
\end{equation}

By Eq. \eqref{action}, to prove the continuity of the map  $t\mapsto H_{\gamma(t)}\psi(t)$ in $L^2(\RE)$ it is enough to show the continuity of $\|\phi''(t)\|$. As the continuity of $U(t)\phi_0 $ is obvious, we just need to show that 
\[\lim_{\de\to 0}\|\hat{\tilde \phi}(t+\de) - \hat{\tilde \phi}(t)\|^2_{L^2(\RE, |k^4|dk)} =0 . \] 
By Eqs. \eqref{tildephi1} and \eqref{tildephi2}, this is reduced to show that 
\[\lim_{\de\to 0}\int_\RE dk \, k^2 \left| \int_t^{t+\delta} ds\,  e^{-ik^2 s}  \dot q(s)\right|^2 = 0. \] 
For the proof of this statement we refer to the proof of Prop. 3.3 in \cite{cacciapuoti:2015ab}.  \hfill $\Box$

%
%
%

\end{document}